\documentclass[aps,final,twocolumn]{revtex4-1}
\usepackage{amsmath,amssymb,amsfonts,amsthm,color,graphicx}

\newcommand{\commentout}[1]{}
\makeatother

\definecolor{nblack}{rgb}{0,0,0}
\definecolor{nblue}{rgb}{0.2,0.2,0.7}
\definecolor{nred}{rgb}{0.7,0.2,0}

\newtheorem*{definition}{Definition}

\newtheorem{theorem}{Theorem}

\begin{document}

\title{Experimental refutation of a class of $\psi$-epistemic models}
\author{M. K. Patra$^{1}$, L. Olislager$^{2}$, F. Duport$^{2}$, J. Safioui$^{2}$, S. Pironio$^{1}$, and S. Massar$^{1}$}
\affiliation{$^{1}$Laboratoire d'Information Quantique, CP 225, Universit{\'e} libre de Bruxelles,\\Av. F. D. Roosevelt 50, B-1050 Bruxelles, Belgium\\
$^{2}$OPERA-Photonique, CP 194/5, Universit{\'e} libre de Bruxelles,\\Av. F. D. Roosevelt 50, B-1050 Bruxelles, Belgium}
\begin{abstract}
The quantum state $\psi$ is a mathematical object used to determine the outcome probabilities of measurements on physical systems. Its fundamental nature has been the subject of discussions since the origin of the theory: is it \emph{ontic}, that is, does it correspond to a real property of the physical system? Or is it \emph{epistemic}, that is, does it merely represent our knowledge about the system? Recent advances in the foundations of quantum theory show that epistemic models that obey a simple continuity condition are in conflict with quantum theory already at the level of a single system. Here we report an experimental test of continuous epistemic models using high-dimensional attenuated coherent states of light traveling in an optical fibre. Due to non-ideal state preparation (of coherent states with imperfectly known phase) and non-ideal measurements (arising from losses and inefficient detection), this experiment only tests epistemic models that satisfy additional constraints which we discuss in detail. Our experimental results are in agreement with the predictions of quantum theory and provide constraints on a class of $\psi$-epistemic models.
\end{abstract}
\maketitle

\section{Introduction}

Most quantum textbooks begin their exposition by postulating that to every physical system corresponds an abstract mathematical object --a ray in Hilbert space-- called the \emph{quantum state}. But does the quantum state correspond to a real physical property of the system? A major reason for doubting is that the quantum state cannot be observed directly. Indeed it can only be reconstructed indirectly by carrying out measurements on ensembles of identically prepared systems \cite{2004-Paris,1993-Aharonov}. Alternatively the quantum state could represent only an observer's knowledge of the physical system, rather than a physical reality. Such an epistemic interpretation of the quantum state could provide an intuitive explanation for many quantum phenomena, such as the measurement postulate and wavefunction collapse \cite{2002-Caves,2003-Brukner,2007-Spekkens}.

Harrigan and Spekkens formulated with precision the above alternatives \cite{2010-Harrigan}. Their starting point is the assumption that every quantum system possesses a real physical state, generally called the ontic state, denoted $\lambda$. The ontic state determines the probabilities of measurement outcomes. When an ensemble of  systems is prepared, different members of the ensemble may be in different ontic states $\lambda$. A preparation procedure $Q$ is therefore associated to a probability distribution $P(\lambda|Q)$ over the ontic states. When a measurement is carried out on a system in ontic state $\lambda$, the probability to obtain outcome $r$ is $P(r|M,\lambda)$. Therefore if preparation $Q$ is followed by measurement $M$ the probability of outcome $r$ is
\begin{equation}
P(r|M,Q) = \sum_{\lambda}  P(r|M,\lambda) P(\lambda|Q) \, .
\end{equation}
These models reproduce the predictions of quantum theory if $P(r|M,Q) = \langle\psi_Q|\mathcal{M}_r|\psi_Q\rangle$, where $\mathcal{M}_r$ is the quantum operator corresponding to the outcome $r$  and $\psi_Q$ is the quantum state assigned by quantum theory to the preparation $Q$.

Following \cite{2010-Harrigan} one can distinguish two classes of models. A model is $\psi$\emph{-ontic} if the preparation of distinct pure quantum states always gives rise to distinct real states. That is, for every $\lambda$ either $P(\lambda|Q) = 0$ or $P(\lambda|Q') = 0$ if the preparations $Q$ and $Q'$ correspond to different quantum states $|\psi_Q\rangle \neq |\psi_{Q'}\rangle$. Every real state $\lambda$ is thus compatible with a unique pure quantum state. The quantum state is ``encoded'' in $\lambda$ and we can consider it to represent a real property of the system.

A model is $\psi$\emph{-epistemic} if the preparation of distinct pure quantum states may result in the same real state $\lambda$. That is, there exist preparations $Q$ and $Q'$ corresponding to non-identical quantum states $|\psi_Q\rangle \neq |\psi_{Q'}\rangle$ such that both $P(\lambda|Q) > 0$ and $P(\lambda|Q') > 0$ for some $\lambda$. In this case, the quantum state is not uniquely determined by the underlying real state.

A consistent formulation of $\psi$-epistemic models would constitute a conceptual revolution of quantum mechanics. However it was recently shown that $\psi$-epistemic models that obey certain natural conditions cannot reproduce the predictions of quantum theory. Pusey, Barrett, and Rudolph (PBR) showed that $\psi$-epistemic models cannot reproduce the predictions of quantum theory if they satisfy the property, termed \emph{preparation independence}, that independently prepared pure quantum states correspond to product distributions over ontic states \cite{2012-Pusey}. A similar result was obtained in \cite{2012-Patra} using a simple argument that relies on a natural assumption of continuity. 
Furthermore, the approach of \cite{2012-Patra} shows that even for a single quantum system $\psi$-epistemic models are strongly constrained. These results show that $\psi$-epistemic models that reproduce the predictions of quantum theory must be both strongly discontinuous and assign a collective ontic state to independently prepared systems. A toy $\psi$-epistemic model that reproduces the predictions of quantum theory and that satisfies these constraints was exhibited in \cite{2012-Lewis}. Additional recent theoretical results presenting no-go theorems for classes of $\psi$-epistemic models have been reported in \cite{2012-Schlosshauer,2013-Hardy,2013-Miller,2012-Maroney,2012-Leifer,2013-Aaronson}.

All these no-go theorems suggest the possibility of novel experiments in the foundations of quantum mechanics. Here we report an experiment that tests the existence of the continuous $\psi$-epistemic models studied in \cite{2012-Patra}. Simultaneously with our work, an experimental test of $\psi$-epistemic models based on the original PBR no-go theorem using two ions in the same trap was reported in \cite{2012-Nigg}.

The interest of basing the experimental test on the no-go theorem of \cite{2012-Patra} is that continuous $\psi$-epistemic models make predictions that are in conflict with quantum theory already at the level of a single system. This greatly facilitates experimental tests. (Constraints on $\psi$-epistemic models at the level of a single system were derived independently in \cite{2013-Hardy}, see the appendix of \cite{2012-Patra} for the relation between the two approaches).

Our experiment uses attenuated coherent states of light traveling in an optical fibre. The coherent state is decomposed into time bins, which provide convenient realisations of high-dimensional Hilbert spaces. 
Specifically we test the predictions of epistemic models in dimensions $3,\,10,\,30,\,50,\,80$. Due to non-ideal state preparation (uncontrolled phase drifts of the laser used) and non-ideal measurements (losses and detector inefficiency), our experiment suffers from ``loopholes''. That is, the observed data could be explained by continuous epistemic models that exploit these loopholes. We argue that if the continuous epistemic model satisfies some reasonable additional hypotheses, then it cannot exploit these loopholes. For instance in the case of detection loophole, we have to make an assumption similar to the fair-sampling assumption often made in non-locality experiments. And in the case of uncertainty in the preparation procedure, we must make the hypothesis that the ontic state does not depend on whether or not control measurements are included in the preparation procedure. Our experimental results are in agreement with the predictions of quantum theory and provide constraints on $\psi$-epistemic models that satisfy these additional constraints. More precisely, we parameterize $\psi$-epistemic models by two parameters, $\delta_0$ that describes how continuous the model is, and $\epsilon$ that describes how epistemic it is. Our experimental results rule out a large class of models labeled by these two parameters.

\section{No-go theorem for\\continuous $\psi$-epistemic models}

As our experiment is based on \cite{2012-Patra}, we recall the relevant results. We first define continuous $\psi$-epistemic models (see fig.~\ref{fig:model} for a depiction of the difference between $\psi$-ontic and $\delta$-continuous $\psi$-epistemic models, and the relevant geometry of the Hilbert space).

\begin{figure}
\begin{center}
\includegraphics[scale=.45]{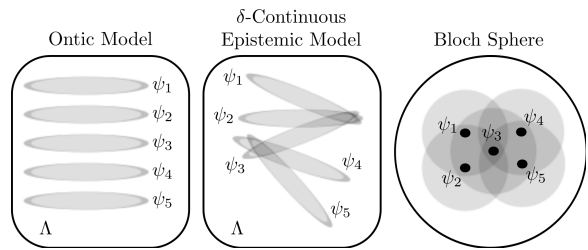}
\caption{Illustration of $\psi$-ontic and $\delta$-continuous $\psi$-epistemic models. Depicted (left and center) is the space $\Lambda$ of ontic states, as well as the support of the probability distribution $P(\lambda|Q_k)$ for preparation $Q_k$ associated to distinct pure states $\psi_k$, $k=1,\ldots,5$. In $\psi$-ontic models (left) distinct quantum states give rise to probability distribution $P(\lambda|Q_k)$ with no overlap. In $\delta$-continuous $\psi$-epistemic models (center),  states that are close to each other (such as $\{\psi_1,\psi_2,\psi_3\}$ and $\{\psi_3, \psi_4,\psi_5\}$) all share common ontic states. However states that are further from each other (such as $\psi_1$ and $\{\psi_4,\psi_5\}$) do not necessarily have common ontic states $\lambda$. On the right we represent the relationship between the states in Hilbert space. We have depicted by dots the positions of the states on the Bloch sphere, and in grey around each state the sphere of radius $\delta$. For instance the sphere of radius $\delta$ around state $\psi_1$ includes states $\psi_2$ and $\psi_3$, but not states $\psi_4$ and $\psi_5$. In $\delta$-continuous epistemic models all states in the ball of radius $\delta$ share at least one ontic state.
}
\label{fig:model}
\end{center}
\end{figure}

\begin{definition} [$\delta$-continuity \cite{2012-Patra}]
Let $\delta>0$ and let $B_\psi^\delta$ be the ball of radius $\delta$ centred on $|\psi\rangle$, i.e., $B_\psi^\delta$ is the set of states $|\phi\rangle$ such that
\begin{equation}
|\langle\phi|\psi\rangle| \geq 1-\delta \, .
\label{eq:delta}
\end{equation}
We say that a model is $\delta$-continuous if for any preparation $Q$, there exists an ontic state $\lambda$ (which can depend on $Q$) such that for all preparations $Q'$ corresponding to quantum states $|\phi_{Q'}\rangle$ in the ball $B^\delta_{\psi_Q}$ centred on the state $|\psi_Q\rangle$, we have $P(\lambda|Q')>0$.
\end{definition}

In order to motivate this definition, recall that according to the definition of $\psi$-ontic and $\psi$-epistemic models, we assign an ontic status to $\psi$ if a variation of $\psi$ necessarily implies a variation of the underlying reality $\lambda$, and we assign it an epistemic status if a variation of $\psi$ does not necessarily imply a variation of the reality $\lambda$. It is then natural to assume a form of continuity for $\psi$-epistemic models: a slight change of $\psi$ induces a slight change in the corresponding ensemble of $\lambda$'s, in such a way that at least some $\lambda$'s from the initial ensemble will also belong to the perturbed ensemble. The above is a slightly stronger form of continuity which asserts that there are real states $\lambda$ in the initial ensemble that will remain part of the perturbed ensemble, no matter how we perturb the initial state, provided this perturbation is small enough.

If $\psi$-epistemic models reproduce the predictions of quantum theory, then there is a fundamental constraint on their $\delta$-continuity. This constraint holds even at the level of a single quantum system.

\begin{theorem} [No-go theorem for $\delta$-continuous models \cite{2012-Patra}]
$\delta$-continuous $\psi$-epistemic models with $\delta \geq 1-\sqrt{(d-1)/d}$ cannot reproduce all the measurement statistics of quantum states in a Hilbert space of dimension $d$.
\end{theorem}

We give the proof of this result, as the construction used is the basis for the experimental test reported below.

\begin{proof}
Consider $d$ preparations $Q_k$ ($k=1,\ldots,d$) corresponding to distinct quantum states $|\psi_k\rangle$ all contained in a ball of radius $\delta$. By definition of a $\delta$-continuous model, there is at least one $\lambda$ for which $\min_{k}P(\lambda|Q_k) > 0$ and thus
\begin{equation}
\epsilon \equiv \sum_\lambda \min_{k}P(\lambda|Q_k) > 0
\label{eq:epsilon}
\end{equation}
(this last quantity can be viewed as a measure of the extent to which distributions over real states overlap in the neighbourhood of a given quantum state). Consider now a measurement $M$ that yields one of $d$ possible outcomes $r=1,\ldots,d$. A $\delta$-continuous model then makes the prediction
\begin{eqnarray} 
\sum_{k}P(k|M,Q_k)
&=& \sum_{k}\sum_{\lambda}P(k|M,\lambda)P(\lambda|Q_k) \nonumber \\
&\geq& \sum_{k}\sum_{\lambda}P(k|M,\lambda)\min_{k}P(\lambda|Q_k) \nonumber \\
&=& \sum_{\lambda}\min_{k}P(\lambda|Q_k) = \epsilon > 0 \, .
\label{eq:nogo}
\end{eqnarray}
According to quantum theory, however, there exist states in a Hilbert space of dimension $d$ contained in a ball of radius $\delta = 1-\sqrt{(d-1)/d}$ such that the left-hand side of eq.~(\ref{eq:nogo}) is equal to $0$. To show this, let $\{|j\rangle \,:\, j=1,\ldots,d\}$ be a basis of the Hilbert space. Consider the $d$ distinct states $|\psi_k\rangle = \frac{1}{\sqrt{d-1}}\sum_{j\neq k}|j\rangle$. These states are all at mutual distance
\begin{equation}
\delta = 1-|\langle\psi_k|\psi\rangle| = 1-\sqrt{(d-1)/d} = \frac{1}{2d} + O\left(\frac{1}{d^2}\right)
\end{equation}
from the state $|\psi\rangle = \frac{1}{\sqrt{d}}\sum_{j}|j\rangle$. Let the measurement $M$ be the measurement in the basis $\{|j\rangle\}$. Then $P(k|M,Q_k) = 0$ for all $k=1,\ldots,d$ and thus $\sum_{k}P(k|M,Q_k) = 0$.
\end{proof}

The above definition and theorem lead us to define a class of $\psi$-epistemic models whose existence can be tested experimentally. They are labeled by the two parameters, $\delta$ and $\epsilon$, that come up in the key equations (\ref{eq:delta}) and (\ref{eq:epsilon}).

\begin{definition} [$\delta\epsilon$-$\psi$-epistemic models]
A $\delta\epsilon$-$\psi$-epistemic model is such that for any set of preparations $Q_k$ corresponding to distinct quantum states $|\psi_k\rangle$ all contained in a ball of radius $\delta$,
\begin{equation}
\sum_\lambda \min_{k}P(\lambda|Q_k) \geq \epsilon \, .
\end{equation}
\end{definition}

\section{Experimental setup}
\label{sec:setup}

In our experiment we realise good approximations of the states $|\psi_{k}\rangle$ using coherent states of light traveling in an optical fibre. The basis states $|j\rangle$ correspond to a photon localised at equally spaced positions in the optical fibre. Such states are called \emph{time bins} as they are conveniently labeled by the time $t_{j}=j\tau$ at which they are detected, where $\tau$ is the spacing (in time) between the time bins, taken to be much larger than the time resolution of the single-photon detectors. (Time bins have been extensively used in quantum optics, and in particular in quantum key distribution, see e.g. \cite{1998-Ribordy,2005-Stucki}). We use up to 80 time bins, leading to a very sensitive experiment, since as follows from proof of theorem~1, the bound on the continuity parameter $\delta$ decreases when the dimensionality $d$ increases (although this conclusion must be somewhat tempered, see section~\ref{sec:results}). We then measure the time of arrival of the photon, which tells us in which bin $k$ the photon is present, and thus provides us with an experimental value for the quantity $\sum_k P(k|M,Q_k)$ which appears on the left-hand side of eq.~(\ref{eq:nogo}).

Our experimental setup is schematised in fig.~\ref{fig:setup}. It is realised using fibre-pigtailed components in the telecommunication C-band. The laser source (Koheras AdjustiK) continuously emits 1\,mW of power at 1549.4\,nm into an optical fibre. Its narrow spectral linewidth ($\Delta\nu\sim1\,$kHz when measured during $120\,\mu$s, as specified by the manufacturer) corresponds to a coherence time $\tau_{\mathrm{coh}}=\left(2\pi\Delta\nu\right)^{-1}\sim160\,\mu$s significantly longer than all other time scales in the experiment. (An upper bound on the laser linewidth was obtained in our group: this laser was previously used in \cite{2010-Leo} and success of this experiment required that the linewidth of the laser be at most $\Delta\nu \leq 20\,$kHz). Power fluctuations of the laser are less than 0.4\,dB in 10\,hours, and side lobes are rejected by more than 63\,dB. This guarantees that within the time interval used to produce the train of pulses the source emits a coherent state of well-defined photon mean number. We then create a pulse train from the cw laser output using an acousto-optic modulator (AOM). Trains of $d=3,\,10,\,30,\,50,\,80$ pulses with one missing are thus created. The AOM (model M111-2J-FxS from Gooch and Housego) has a 25\,ns rise/fall time and a 50\,dB extinction ratio in cw regime. In the pulsed regime used in the experiment the extinction ratio is estimated to be $\mathit{Ext}=(40\pm1.5)\,$dB. A pattern generator (Hewlett-Packard model 81110A) drives the AOM. The AOM is turned on for 100\,ns and then off for 200\,ns. To attenuate the state we use a fixed optical attenuator complemented with a variable one (Hewlett-Packard model 8156A). An overall attenuation of up to 120\,dB with an absolute precision of 0.1\,dB and a repeatability of 0.01\,dB can be achieved. The mean number of photons in the pulse train is $\langle n \rangle = 0.2$. To complete the state preparation phase, the light is sent through a fibre spool long enough to store the complete pulse train. The photons are detected with a superconducting single-photon detector (SSPD from Scontel) cooled to $(1.7\pm0.1)\,$K with overall efficiency (including losses in optical components after the state preparation and data acquisition inefficiency) $\eta=(4\pm0.2)\,\%$ and dark-count rate $\mathit{Dk}=(3\pm1)\,$Hz. The dark-count rate can be measured with high precision, but it is sensitive to environmental conditions (temperature of the detectors, ambient light) that fluctuate during the experiment, which is the reason the quoted error is large. The data acquisition is realised by a time-to-digital-converter (Agilent Acqiris system). The overall time resolution of the detector and data acquisition is approximately 150\,ps. In order to minimise the effects of the finite rise and fall time of the AOM, we only keep the clicks that occur during an interval of width $T_p=80\,$ns centred on the middle of each time bin. For the different values of dimension 
 $d=3,\,10,\,30,\,50,\,80$ studied, the total number of times each state $\alpha_k$ was produced was $12,\,10,\,4,\,3,\,2 \times 128 \times 10^4$ respectively. The data was acquired over the duration of one week. An example of recorded data is depicted in fig.~\ref{fig:setup}.

\begin{figure}
\begin{center}
\includegraphics[scale=.45]{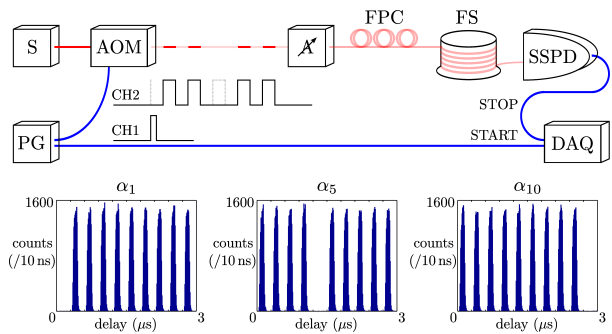}
\caption{Top panel: Experimental setup. Red links are optical fibres and blue links are radio-frequency cables. A continuous laser source (S) emits a coherent state whose intensity is modulated by an acousto-optic modulator (AOM) driven by a pattern generator (PG), yielding a train of $d$ pulses with one missing. The train of pulses is then attenuated to the single-photon regime by passing through an optical attenuator (A). The complete pulse train is stored in a 5\,km fibre spool (FS), and then sent to a superconducting single-photon detector (SSPD). A fibre polarization controller (FPC) is used to ensure maximum sensitivity of the SSPD (which is polarization sensitive). When a photon is detected, an electric signal is sent to a data acquisition system (DAQ) which stores the time at which the detection event has taken place relative to the time at which the state preparation began. Bottom panel: Example of experimental results for $d=10$, and states $|\alpha_{1}\rangle$, $|\alpha_{5}\rangle$, $|\alpha_{10}\rangle$. Horizontal axis: time at which a detection is registered. Vertical axis: number of detections in each 10\,ns time window. Continuous $\psi$-epistemic models would predict a non-zero count rate in the bins which should be empty. The small number of counts which do occur due to detector dark counts and finite extinction ratio of the AOM are not visible on this scale.}
\label{fig:setup}
\end{center}
\end{figure}

\section{Interpretation of the experiment}

In the experiment described above, a long coherence laser operating well above threshold produces at its output a coherent state \cite{2008-Garrison} which is cut into a train of $d$ pulses with one missing, and then attenuated. 
Using the standard notation of quantum optics, this procedure yields the coherent state
\begin{equation}
|\alpha_{k}\rangle \simeq |0\rangle + \alpha \left(\frac{1}{\sqrt{d-1}}\sum_{j\neq k}|j\rangle\right) + O\left(\alpha^2\right) \, ,
\label{eq:coherent}
\end{equation}
where $|0\rangle$ denotes the vacuum state and where for simplicity of notation we omit the contributions from two and more photons (this does not affect the interpretation of the experiment, see below). The mean number of photons is chosen to be $\langle n \rangle = \alpha^2 = 0.2$ for all dimensions $d=3,\,10,\,30,\,50,\,80$ investigated.

Because the states prepared in our experiment have a significant vacuum component, and because of losses and finite detector efficiency, the preparation of a quantum state $|\alpha_k\rangle$ can either give rise to a detection in one of the time bins, or to a \emph{no-click} event wherein no photon is registered. These no-click events affect the interpretation of the experiment.

Furthermore the expression eq.~(\ref{eq:coherent}) supposes an ideal output from the laser source. However lasers fluctuate. For a laser operating well above threshold, the major source of fluctuation is phase drift (responsible for the finite linewidth of the laser). Even though the coherence time of the laser $\tau_{\mathrm{coh}}\sim160\,\mu$s is much longer than the longest pulse train (of length $80\times300\,\mathrm{ns}=24\,\mu$s), we cannot neglect this phase drift. A more precise description of the preparation procedure is therefore that it yields the state
\begin{eqnarray}
|\alpha_{k,\varphi}\rangle
&=& \exp\left[-\alpha^2/2\right] \prod_{j\neq k} \exp\left[\frac{\alpha e^{i\varphi_j} a^\dagger_j}{\sqrt{d-1}}\right] |0\rangle
\label{eq:alphaphi} \\
&\simeq& |0\rangle + \alpha \left(\frac{1}{\sqrt{d-1}}\sum_{j\neq k}e^{i\varphi_j}|j\rangle\right) + O\left(\alpha^2\right) \, ,
\label{eq:coherentphase}
\end{eqnarray}
where $y_j=\varphi_j-\varphi_{j-1}$ should be modeled as independent Gaussian random variables, since the phase fluctuations of a laser are generally modeled as a random walk of the phase \cite{2008-Garrison}.

We now discuss the consequences of imperfect state preparation and measurement for the interpretation of the experiment.

\subsection{Detection loophole}
\label{sec:detection}

As said above, because the states prepared in our experiment have a significant vacuum component, and because of losses and finite detector efficiency, the preparation of a quantum state $|\alpha_k\rangle$ can either give rise to a detection in one of the time bins, or to a \emph{no-click} event when no photon is registered. The latter are significantly more common than the detections. Indeed, use of coherent states with $\langle n \rangle = 0.2$ and overall detection efficiency of approximately $4\%$ (see section~\ref{sec:setup}) yields an overall probability of registering a click of approximately $P(\mathrm{clk}) = 8\cdot10^{-3}$, where $\mathrm{clk}$ is the event that the detector clicks in one of the time bins.

These no-click events affect the interpretation of the experiment. To understand why, remember that the key point of theorem~1 was showing that if there is an ontic state that occurs with positive probability for all states $|\psi_k\rangle$, then one finds a contradiction with quantum theory. But in the presence of no-click events this contradiction no longer holds. Indeed there exists a trivial $\psi$-epistemic model that explains our experimental results in which the ontic states common to all preparations $|\alpha_k\rangle$ only give rise to no-click events. Furthermore, the vacuum component of the states $|\alpha_k\rangle$ affects the interpretation of their mutual scalar product (since states that are almost orthogonal become arbitrarily close to each other when superposed with a sufficiently large vacuum component that does not contribute to the click events).

The basis for generalising the analysis is to distinguish between two classes of ontic states: the set of ontic states denoted $\Lambda_0$ which only give rise to no-click events; and the complementary set $\Lambda_\mathrm{clk} = \Lambda \setminus \Lambda_0$, see fig.~\ref{fig:lambda}. All ontic states belonging to the set $\Lambda_\mathrm{clk}$ give rise to a click with positive probability. If for each preparation $|\alpha_k\rangle$, ontic states belonging to $\Lambda_\mathrm{clk}$ occur with positive probability, then we can apply an analog of theorem~1.

\begin{figure}
\begin{center}
\includegraphics[scale=.45]{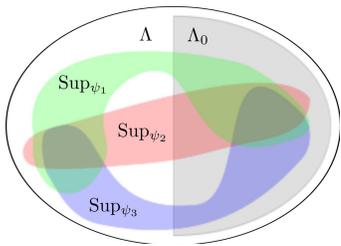}
\caption{Schematic depiction of the structure of the space $\Lambda$ of ontic states in the case of inefficient detectors. We have represented the supports Sup$_{\psi_k}$ of the probability distributions $p(\lambda|\psi_k)$ for three states $|\psi_1\rangle,|\psi_2\rangle,|\psi_3\rangle$. The subset of ontic states which never give rise to a click is denoted by $\Lambda_0$. The experiment described in the main text cannot rule out the existence of a non-empty intersection of the supports Sup$_{\psi_k}$ and of $\Lambda_0$, since all ontic states in this intersection always give rise to no-click events. The reported experiment can rule out the existence of a non-empty intersection of the supports Sup$_{\psi_k}$ and the complementary space $\Lambda_\mathrm{clk} = \Lambda \setminus \Lambda_0$.}
\label{fig:lambda}
\end{center}
\end{figure}

To proceed quantitatively, we first redefine the notion of distance between states to take into account that the vacuum component never gives rise to a click. The new notion of distance $\delta_0$ should have the following properties: (1) it measures the distance between states on the space orthogonal to the vacuum state; and (2) it equals the old distance $\delta_0 = \delta$ on the single-photon space. The exact way $\delta_0$ acts on the two and more photon space is not essential for the argument (since the overlap of the states we consider with the two and more photon space is small). The reason for property (1) is that the vacuum component of the state will not give rise to a click, and hence does not give rise to any measurable quantity. The reason for property (2) is that in the case of single-photon states and perfect detectors we wish to recover the notion of $\delta$-continuity defined above.

\begin{definition} [$\delta_0$-continuity]
A model with no-click events is $\delta_0$-continuous if, for all preparations $Q$ corresponding to pure state $|\psi_Q\rangle$ with $P(\mathrm{clk}|\psi_Q)>0$, there exists an ontic state $\lambda$ (which can depend on $|\psi_Q\rangle)$ such that, for all preparations $Q'$ corresponding to state $|\phi_{Q'}\rangle$ with
\begin{equation}
|\langle\tilde\phi_{Q'}|\tilde\psi_Q\rangle| \geq 1-\delta_0 \, ,
\label{eq:delta0}
\end{equation}
we have $P\left(\lambda|{\phi_{Q'}},\mathrm{clk}\right)>0$, where
\begin{equation}
|\tilde\psi_{Q}\rangle = \frac{\left(I-|0\rangle\langle 0|\right) |\psi_Q\rangle}{\left|\left(I-|0\rangle\langle 0|\right) |\psi_Q\rangle\right|}
\label{eq:tilde}
\end{equation}
is the projection of $|\psi_Q\rangle$ onto the space orthogonal to the vacuum, and $|\tilde\phi_{Q'}\rangle$ is similarly defined.
\end{definition}

As illustration of this new notion of distance, consider the coherent states
\begin{equation}
|\alpha_k\rangle = \exp\left[-\alpha^2/2\right] \exp\left[\frac{\alpha}{\sqrt{d-1}}\sum_{j\neq k}a^\dagger_j\right] |0\rangle \, ,
\label{eq:alphak}
\end{equation}
and the reference state
\begin{equation}
|\alpha_0\rangle = \exp\left[-\alpha^2/2\right] \exp\left[\frac{\alpha}{\sqrt{d}}\sum_{j}a^\dagger_j\right] |0\rangle \, .
\label{eq:alpha0}
\end{equation}
Then we have
\begin{equation}
|\langle\tilde\alpha_k|\tilde\alpha_{0}\rangle| = \frac{e^{\alpha^2\sqrt{(d-1)/d}}-1}{e^{\alpha^2}-1} \, .
\label{eq:deltatilde}
\end{equation}
Therefore the distance $\delta_0$ is given by
\begin{eqnarray} 
\delta_0(d,\alpha^2)
&=& 1-\frac{e^{\alpha^2\sqrt{(d-1)/d}}-1}{e^{\alpha^2}-1}
\nonumber \\
&=& 1-\sqrt{\frac{d-1}{d}} + O\left(\alpha^2\right) \, ,
\label{eq:deltaexpdalpha}
\end{eqnarray}
where we have expanded to leading order in $\alpha^2$. For the experimentally relevant case $\alpha^2 = \langle n \rangle = 0.2$, we find $\delta_0(d,\alpha^2=0.2)\simeq 0.55/d$.

We must also introduce the notion of inefficient detectors.

\begin{definition} [Inefficient detector]
An inefficient detector provides a response that depends only on the photon number. If there is no photon, it does not click. If there are one or more photons, the probability of clicking is strictly positive.
\end{definition}

We now consider the analog of theorem~1 in the case of inefficient detectors. We prove it using coherent states, as these are the relevant ones for our experiment, but it extends trivially to other states.

\begin{theorem} [No-go theorem for $\delta_0$-continuous models]
$\delta_0$-continuous $\psi$-epistemic models with $\delta_0 > 1-\sqrt{(d-1)/d}$ cannot reproduce all the measurement statistics on coherent states of $d$ modes, even in the presence of inefficient detectors.
\end{theorem}

\begin{proof}
Fix $\delta_0$. Consider $d$ preparations $Q_k$ ($k=1,\ldots,d$) corresponding to coherent states $|\alpha_k\rangle$ which are at distance $1-|\langle \tilde \alpha_k|\tilde\alpha_0 \rangle| > \delta_0$ from some reference coherent state $|{\alpha_0}\rangle$. We suppose that for all $k$, $P(\mathrm{clk}|\alpha_k)>0$. We define the subnormalised measure $\omega_{\mathrm{clk}}(\lambda)=\min_{k}P(\lambda|{\alpha_k},\mathrm{clk})$, with $P(\lambda|{\alpha_k},\mathrm{clk})=\frac{P(\mathrm{clk}|\lambda)P(\lambda|{\alpha_k})}{P(\mathrm{clk}|{\alpha_k})}$ given by Bayes rule. Because of $\delta_0$-continuity we have
\begin{equation}
\epsilon_0 = \min_\lambda \omega_{\mathrm{clk}}(\lambda)>0 \, .
\end{equation}
We then have the mathematical identity:
\begin{eqnarray}
\sum_{k} P(k|{\alpha_k},\mathrm{clk})
&=& \sum_{k} \frac{P(k|{\alpha_k})}{P(\mathrm{clk}|{\alpha_k})}
\nonumber \\
&=& \sum_{k} \sum_{\lambda}\frac{P(k|\lambda)}{P(\mathrm{clk}|\lambda)}\frac{P(\mathrm{clk}|\lambda)P(\lambda|{\alpha_k})}{P(\mathrm{clk}|{\alpha_k})}
\nonumber \\
&=& \sum_{k}\sum_{\lambda}\frac{P(k|\lambda)}{P(\mathrm{clk}|\lambda)}P(\lambda|{\alpha_k},\mathrm{clk})
\nonumber \\
&\geq& \sum_{k}\sum_{\lambda}\frac{P(k|\lambda)}{P(\mathrm{clk}|\lambda)}\omega_{\mathrm{clk}}(\lambda)
\nonumber \\
&=& \sum_{\lambda}\omega{}_{\mathrm{clk}}(\lambda) = \epsilon_0 > 0 \, .
\label{eq:pclick}
\end{eqnarray}
Note that the left-hand side of eq.~(\ref{eq:pclick}) is given experimentally by
\begin{equation}
\epsilon_\mathrm{expt} = \sum_{k} P(k|Q_k,\mathrm{clk}) = \sum_{k} \frac{N(k,Q_k)}{\sum_{j}N(j,{Q_k})} \, ,
\label{eq:epsexp}
\end{equation}
where $N(j,Q_k)$ is the number of clicks registered in outcome $j$ when one prepares state $\alpha_k$ and $\mathrm{clk}$ is the event that the detector clicks. If we take the states $|{\alpha_j}\rangle$ to be given by eq.~(\ref{eq:alphak}) and the state $|{\alpha_0}\rangle$ to be given by eq.~(\ref{eq:alpha0}), then by taking the parameter $\alpha^2\to 0$, we can take $\delta_0$ in the above argument arbitrarily close to $1-\sqrt{(d-1)/d}$, see eq.~(\ref{eq:deltaexpdalpha}). Since inefficient detectors will never click if there are zero photons, but have non-zero probability of clicking if there is at least one photon, we will have 
$N(k,Q_k)=0$ for all $k$ and ${\sum_{j}N(j,{Q_k})}>0$ for all $k$, and therefore $\epsilon_\mathrm{expt}=0$, in contradiction with eq.~(\ref{eq:pclick}).
\end{proof}

The above definitions and theorem lead us to define a class of $\psi$-epistemic models whose existence can be tested experimentally, even in the presence of losses and inefficient detectors. These models are labeled by the parameters $\delta_0$ and $\epsilon$.

\begin{definition} [$\delta_0\epsilon$-$\psi$-epistemic models]
Consider an arbitrary number of preparations $Q_k$ corresponding to distinct quantum states $|\psi_k\rangle$ all contained in a ball of radius $\delta_0$, where $\delta_0$ is given by equations (\ref{eq:delta0}) and (\ref{eq:tilde}). A $\delta_0\epsilon$-$\psi$-epistemic model is such that, for all choices of $Q_k$,
\begin{equation}
\sum_\lambda \min_{k}P(\lambda|Q_k) \geq \epsilon \, .
\end{equation}
\end{definition}

It is this class of models that are experimentally tested in our experiment.

Note that in appendix~\ref{sec:alternative} we present a second approach to deal with the no-click events, based on an analysis of the measurement process. 
This approach conserves the original $\delta$-continuity condition, but postulates that the predictions of the $\psi$-epistemic model are independent of how the measurement is realised, provided that the realisation gives the same statistics within quantum theory. Using this second approach, we obtain directly a bound on $\min_k P(\lambda|Q_k)$ rather than $\omega_{\mathrm{clk}}(\lambda)$.

\subsection{Preparation of mixed states}
\label{sec:mixed}

The output of the laser used in our experiment fluctuates, giving rise to the finite linewidth of the laser. The most important fluctuations are expected to arise from phase drift which is generally modeled as a random walk of the laser phase \cite{2008-Garrison}. The time scale of this phase drift is related to the laser linewidth by $\tau_{\mathrm{coh}}=\left(2\pi\Delta\nu\right)^{-1}\sim160\,\mu$s. Even though this time scale is much longer than the longest train of pulses (of length $80\times300\,\mathrm{ns}=24\,\mu$s), these fluctuations have an important impact on the interpretation of the experiment.

The consequence of these phase fluctuations is that the states prepared can be modeled as eq.~(\ref{eq:coherentphase}), where $y_j=\varphi_j-\varphi_{j-1}$ are independent identically distributed random variables with normal distribution
\begin{equation}
P(y_j) = \frac{1}{\sqrt{4\pi Dt_0}}e^{-\frac{y_j^2}{4Dt_0}} \, ,
\label{eq:disty}
\end{equation}
where $D=1/\tau_{\mathrm{coh}}$  is the diffusion constant \cite{1997-Scully} and $t_0$ is the time between centers of two time bins.

We take this model as basis for the analysis. Extensions, taking into account for instance phase drift within each time bin, or intensity fluctuations of the laser, are briefly discussed below. In order to understand the implications of fluctuations on the prepared state, we first discuss the interpretation of the experiment if the phases $\varphi_j$ were known. We then consider the  experimentally relevant case where the phases $\varphi_j$ are unknown.

If the phases $\varphi_j$ were known, then for each state $|\alpha_{k,\varphi}\rangle$ we could compute the $\delta_0$-distance to the reference state $|\alpha_0\rangle$. We would then keep only the data corresponding to the case where $\delta_0(\alpha_{k,\varphi},\alpha_0)$ is less than some threshold $\Delta_0$. For the states $|\alpha_{k,\varphi}\rangle$, we have (for an ideal experiment) that $P(k|\alpha_{k,\varphi})=0$, and hence the contradiction in eq.~(\ref{eq:pclick}) holds. For the subset of states with $\delta_0(\alpha_{k,\varphi},\alpha_0)<\Delta_0$, we can estimate the value of $\epsilon_\mathrm{expt}$ (which will be non-zero because of experimental imperfections) from the measurement data. The data then exclude $\delta_0$-continuous epistemic models with $\delta_0>\Delta_0$ and $\epsilon>\epsilon_\mathrm{expt}$.

If we do not know the phases $\varphi_j$, then we must make additional assumptions. Note that if one averages over the unknown phases $\varphi_j$, the state prepared by the device is a mixed state. This is problematic as the notions of $\psi$-epistemic and $\psi$-ontic models are defined for pure states. Therefore the no-go theorems do not apply directly. We reason around this difficulty as follows. There is in principle a simple modification of the experimental procedure that could be used to determine the phase of the laser: namely part of the laser light could be diverted and then measured. This additional measurement has not been carried out. But it is natural to assume that the probability distribution of ontic states $P(\lambda|Q_k)$ does not depend on whether or not these additional measurements are carried out. (We note that in our preparation procedure, a large part of the light is in fact diverted, and then absorbed, by the attenuators).

Making this assumption, we can consider that the state preparation yields pure states of the form eq.~(\ref{eq:coherentphase}) with small, random, phases affecting each time bin. We do not know what are the values of the phases, but we can determine (using numerical Monte-Carlo simulations, or analytical calculations) the probability distribution of $\delta_0(\alpha_{k,\varphi},\alpha_0)$. From this probability distribution we can estimate the probability $q$ that $\delta_0$ is less than a specific value $\Delta_0$: $P[\delta_0(\alpha_{k,\varphi},\alpha_0)\leq\Delta_0] = q(\Delta_0)$. Equivalently we know what proportion $q$ of prepared states had $\delta_0(\alpha_{k,\varphi},\alpha_0)\leq\Delta_0(q)$. We also know the experimentally determined value of $\epsilon_\mathrm{expt}$. However this value is an average over all prepared states. For a conservative estimate we make the worst case assumption that all the contribution to $\epsilon_\mathrm{expt}$ comes from the states with $\delta_0(\alpha_{k,\varphi},\alpha_0)\leq\Delta_0$. This implies that we must make the substitution $\epsilon_\mathrm{expt}\to\epsilon_\mathrm{expt}(q)= \epsilon_\mathrm{expt}/q$. This is the price to pay for not having experimentally measured the phases $\varphi_j$. The data then exclude $\delta_0$-continuous epistemic models with $\delta_0>\Delta_0$ and $\epsilon>\epsilon_\mathrm{expt}(q)$. Note that we can vary the parameters $q$ and $\Delta_0$ to exclude a region as large as possible in the $\delta_0,\epsilon$ plane.

In practice we proceed as follows. We fix the dimension $d=3,\,10,\,30,\,50,\,80$. We fix $k\in\{1,\ldots,d\}$. We choose at random variables $y_j$, $j=1,\ldots,d$, drawn from the distribution eq.~(\ref{eq:disty}). To these variables we associate the state $|\alpha_{k,\varphi}\rangle$ defined in eq.~(\ref{eq:alphaphi}). We then compute the distance $\delta_0(\alpha_{k,\varphi},\alpha_0)=1-|\langle \tilde\alpha_{k,\varphi}|\tilde\alpha_{0}\rangle|$, where $|\alpha_{0}\rangle$ is given by eq.~(\ref{eq:alpha0}). We repeat the procedure $10^6/d$ times for each value of $k$. For simplicity we then average the resulting histograms over $k$, yielding a probability distribution for $\delta_0$: $P(\delta_0)=\frac{1}{d}\sum_{k=1}^d P[\delta_0(\alpha_{k,\varphi},\alpha_0)]$. From this numerically determined distribution we can compute with high precision the function $\Delta_0(q)$ given by $P((\delta_0\leq \Delta_0)=q$.

Finally we note that the states prepared by the laser may differ from the ideal state eq.~(\ref{eq:coherent}) in more ways than are modeled in eq.~(\ref{eq:coherentphase}). Such effects could include intensity fluctuations of the laser, or phase drift within each time bin. We could take them into account by using a better model of the laser output. However since the linewidth of a laser well above threshold is generally modeled as being entirely due to phase drift, and since the coherence time is much longer than the duration of one time bin, we expect that the above takes into account most of the effects due to uncertainty in the state preparation. We note that our procedure of ascribing all the contribution to $\epsilon_\mathrm{expt}$ from the states with $\delta_0(\alpha_{k,\varphi},\alpha_0)\leq\Delta_0$ is very conservative, and implies that the true value of $\epsilon_\mathrm{expt}$ is probably significantly smaller than the one we use.

Note that in appendix~\ref{sec:analytic} we present an analytic estimate of $E\left[\delta\right]$ that coroborates the numerical calculations of $\Delta_0(q)$ outlined above.

\section{Results}
\label{sec:results}

Our raw experimental results are reported in table~\ref{table:results} and fig.~\ref{fig:epsilon1}. Specifically we give the number $d$ of time bins, the measured fraction of clicks in the bin that should contain no photon, i.e. $\epsilon_\mathrm{expt}(d)=\sum_{k}\frac{N(k,Q_{k})}{\sum_{j}N(j,Q_{k})}$, and its statistical error.

\begin{table}[!h]
\centering
\begin{tabular}{c| c c c c c}
$d$ & 3 & 10 & 30 & 50 & 80 \\
\hline 
$\epsilon_\mathrm{expt} \times 10^3$ & $0.26$ & $0.45$ & $1.27$ & $1.62$ & $1.66$ \\
$\mathrm{err} \times 10^3$ & $\pm0.05$ & $\pm0.07$ & $\pm 0.18$ & $\pm 0.23$ & $\pm 0.28$
\end{tabular}
\caption{Experimental results. $d$ is the dimension of the quantum state space. The value of $\epsilon_\mathrm{expt}$, given by eq.~(\ref{eq:epsexp}), is directly measured. The last line gives the statistical uncertainty on $\epsilon_\mathrm{expt}$.}
\label{table:results}
\end{table}

\begin{figure}[!h]
\begin{center}
\includegraphics[scale=.45]{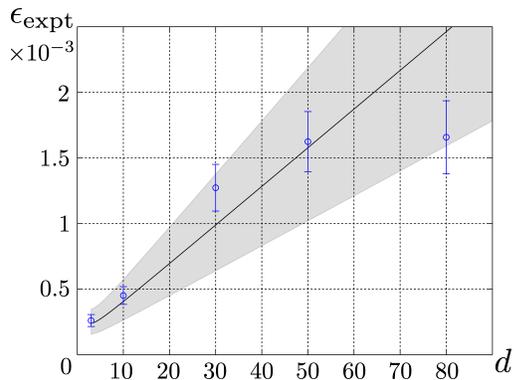}
\caption{Experimental bound on $\psi$-epistemic models as a function of dimension $d$ of the quantum state space. The vertical axis gives the measured value of $\epsilon_\mathrm{expt}=\sum_{k}\frac{N(k,Q_{\alpha_k})}{\sum_{j}N(j,Q_{\alpha_k})}$, where $N(j,Q_{\alpha_k})$ is the number of clicks registered in bin $j$ when one prepares state $\alpha_k$ (error bars are statistical). These values are also given in table~\ref{table:results}. The curve gives the dependency of $\epsilon_\mathrm{expt}$ on $d$ as predicted by quantum theory, taking into account the values of measured experimental parameters. The grey area gives the range in which this theoretical prediction could vary, given the uncertainty on dark-count rate, extinction ratio, and overall detection probability $\eta \langle n \rangle$. The main uncertainty comes from the dark-count rate which depends on the exact temperature of the detectors and the amount of ambient light, both of which can vary during the experiment. Positive deviation from the curve would signal a break-down of quantum theory. The absence of such deviation rules out a large class of $\delta$-continuous models.}
\label{fig:epsilon1}
\end{center}
\end{figure}

The fact that $\epsilon_\mathrm{expt}(d)$ is not strictly zero is expected, since the optical components are imperfect. We have estimated the expected values of $\epsilon_\mathrm{expt}(d)$ from the following measured experimental parameters: extinction ratio of the AOM, mean number of photons in each pulse, optical attenuation and detector efficiency, and detector dark counts. The probability of detecting a photon in bin $k$ when state $\alpha_k$ is prepared is approximately $\mathit{Dk} \, T_p + \mathit{Ext} \, \langle n\rangle \eta/(d-1)$, and the probability of detecting a photon when state $\alpha_k$ is prepared is approximately $\langle n\rangle \eta$, where $d$ is the dimension of the state, $\mathit{Dk} \, T_p$ is the probability of a dark count during a pulse, $\mathit{Ext}$ is the extinction ratio of the AOM, $\langle n \rangle$ is the mean number of photons in the pulse train, and $\eta$ is the overall detection efficiency. The ratio of these two quantities multiplied by $d$ yields the following estimate for the experimentally measured quantity $\epsilon_\mathrm{expt}$:
\begin{equation}
\mbox{Expected value of } \epsilon_\mathrm{expt} = \frac{\mathit{Dk}\,T_p}{\langle n\rangle \eta} d +\mathit{Ext}\frac{d}{d-1} \, .
\end{equation}
This expected value, including its uncertainty, is plotted in grey in fig.~\ref{fig:epsilon1}. Deviations from this expected behaviour could signal that quantum theory should be replaced by an epistemic model. The measured values of $\epsilon_\mathrm{expt}$, which are of the order $10^{-3}$, do not exhibit large deviations from the expected behaviour of $\epsilon_\mathrm{expt}(d)$.

These experimental results can be used to rule out a class of $\delta_0\epsilon$-$\psi$-epistemic models (see definition at the end of section~\ref{sec:detection}). Specifically we proceed as follows. Using the  procedure outlined at the end of section~\ref{sec:mixed}, we can determine for each dimension $d$ the function $\Delta_0(q)$. Specifically we choose a series of values of $0<q<1$, and compute the corresponding value of $\Delta_0(q)$. Then, for each $d$ and for each of these values of $q$, the models with $\delta_0 \geq \Delta_0(q)$ and $\epsilon \geq  \epsilon_\mathrm{expt}(d)/q$ are ruled out. Thus for each dimension $d$, we rule out a region in the $\delta_0,\epsilon$ plane. These results are given in fig.~\ref{fig:epsilon2}.

\begin{figure}[!h]
\begin{center}
\includegraphics[scale=.45]{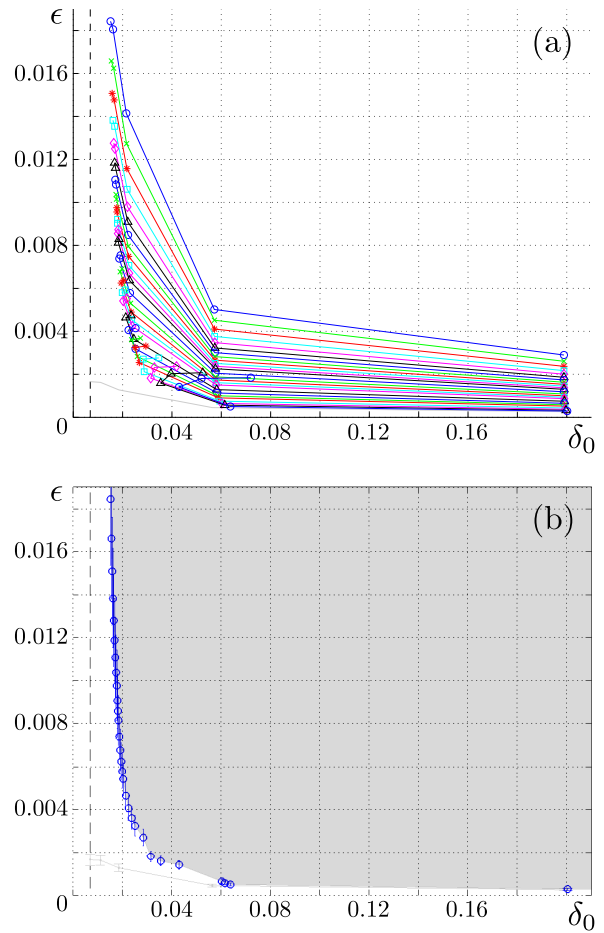}
\caption{(a) Experimentally excluded region in the $\delta_0,\epsilon$ plane. For each choice of the parameter $q$ we compute the values $\Delta_0(d,q)$ and $\epsilon_\mathrm{expt}(d)/q$ which are then plotted in the figure, and connected by a line. For curves from top to bottom, $q=\,$ 0.09, 0.10, 0.11, 0.12, 0.13, 0.14, 0.15, 0.16, 0.17, 0.18, 0.19, 0.20, 0.22, 0.24, 0.26, 0.28, 0.30, 0.35, 0.40, 0.45, 0.50, 0.60, 0.70, 0.80, 0.90. The couples $(\delta_0,\epsilon_\mathrm{expt}(d))$ that would be obtained if one did not take into account the phase fluctuations of the laser are plotted in light grey (lowest curve). For $d=3,\,10$, taking into account the phase fluctuations of the laser does not modify significantly the results, while for $d=30,\,50,\,80$ the effect is important. In fact the excluded region for $d=80$ is practically the same as for $d=50$ when the phase fluctuations are taken into account. (b) Here we take only the points that are most constraining. We plot them with their error bars. The grey zone is the area in the $\delta_0,\epsilon$ plane that is excluded by our experiment.}
\label{fig:epsilon2}
\end{center}
\end{figure}

In fig.~\ref{fig:epsilon2}~(a) we plot all the couples $(\Delta_0(d,q),\epsilon_\mathrm{expt}(d)/q)$ obtained by this procedure. For comparison we also plot the couples $(\delta_0,\epsilon_\mathrm{expt})$ that would be obtained if we did not take into account the phase fluctuations of the laser. For small $d$, taking into account the phase fluctuations has a very small effect on the results because over the duration of $3$ or $10$ time bins, the phase fluctuations have only increased very little the value of $\delta_0$. However when the number of time bins increases, the effect of the phase fluctuations becomes much more important, and this significantly affects the results. In fig.~\ref{fig:epsilon2}~(b) we plot only those couples $(\Delta_0(d,q),\epsilon_\mathrm{expt}(d)/q)$ which are most constraining, and give their statistical error.

\section{Discussion}

Whether the quantum wavefunction is a real physical wave or a summary of our knowledge about a physical system is a question that has divided physicists since the inception of quantum theory. A precise formulation of these two alternatives, opening the way to clear-cut answers, was provided by Harrigan and Spekkens \cite{2010-Harrigan}: if the wavefunction corresponds to a real, ontic, property of physical systems, the preparation of a system in different pure quantum states should always result in different physical states. If, on the other hand, the wavefunction has an epistemic status, such preparations should sometimes result in the same underlying physical state.

Following the breakthrough of PBR \cite{2012-Pusey}, a flurry of no-go theorems for $\psi$-epistemic models obeying natural constraints have recently been proposed \cite{2012-Patra,2012-Schlosshauer,2013-Hardy,2013-Miller,2012-Maroney,2012-Leifer,2013-Aaronson}. These no-go theorems inspire novel experimental tests. Here, we reported  a test of $\psi$-epistemic model based on the argument introduced in \cite{2012-Patra}.  There are two main motivations to perform such experiments. 

First, given that there are good reasons to support an epistemic view of the quantum state \cite{2007-Spekkens}, the no-go theorems provide new directions in which to look for potential deviations from the expected quantum predictions.
Our experimental results do not exhibit any such deviations, therefore strengthening our belief in the validity of quantum theory. 

A second, related, motivation for performing an implementation of the no-go theorems is to rule out experimentally (certain classes of) $\psi$-epistemic models, in the same way that the violations of Bell inequalities rule out locally-causal models. However, while such experiments have some common features with a Bell test, they also differ from it in several ways. To simplify the discussion, let us first consider the case of an ideal experiment free of experimental errors and noise. The proof of Theorem~1 tells us that if we prepare a system according to $d$ possible procedures $Q_k$ and subject it to a measurement $M$, then the observed value $\epsilon_{\mathrm{expt}}=\sum_k P(k|M,Q_k)$ provides a constraint on the extent $\epsilon$ to which the distributions over real states associated to each preparation overlap. In particular, if $\epsilon_{\mathrm{expt}}=0$ then no common real state $\lambda$ can be associated to all the preparations $Q_k$. This conclusion is obtained independently of what specific states (pure or mixed) are used, or what specific measurements are performed. It only depends on the observed measurement statistics, as does the violation of a Bell inequality.

However contrary to Bell inequalities, the observation of a value $\epsilon_{\rm{expt}}=0$ does not per se imply that some ``intrinsically quantum" or ``non-classical" behaviour has been produced in the experiment. Indeed, the quantity $\epsilon_{\rm{expt}}$ has not been introduced to distinguish between a ``classical" and a ``quantum" worldview (as in Bell inequalities), but between $\psi$-ontic and $\psi$-epistemic models. Thus for instance the $\psi$-epistemic model presented in \cite{2012-Lewis} perfectly reproduces all predictions of quantum mechanics.

Furthermore, the statistics of an experimental test of  $\psi$-ontic versus $\psi$-epistemic models could very easily be reproduced using purely classical states. For instance the use of $d$ classical states of the form $\rho_k=(1/d)\,\sum_{j\neq k} |j\rangle\langle j|$ instead of the states $|\psi_k\rangle$ in Theorem~1 would also yield a value of $\epsilon_{\rm{expt}}=0$. However, because these states are not pure, this experiment would not exclude  $\psi$-epistemic models.

This shows that when carrying out an experimental test of  $\psi$-ontic versus $\psi$-epistemic models, the kind of preparation procedures used in the experiment matters to its interpretation. Indeed, a $\psi$-epistemic model does not need to predict different results than quantum theory for all preparation procedures (for instance not for those associated to the purely-classical distributions $\rho_k$ over orthogonal states mentioned above), but only for those corresponding to pure quantum states that are sufficiently close (within distance $\delta$ in the case of $\delta$-continuous models). A meaningful test of $\delta$-continuous $\psi$-epistemic models must therefore be based on two components:  the measurement of  $\epsilon_{\rm{expt}}$ and a reasonable confidence that preparations corresponding to pure quantum states within distance $\delta$ have been used. This should be contrasted with Bell experiments that are "device independent": their interpretation are independent of the details of the state preparation and measurement procedures.

In this later respect, our experiment has some specific weaknesses related to the experimental system used (photonic time bins obtained by chopping and attenuating a cw laser). First, the use of coherent states that have a non-zero vacuum component and inefficient detectors resulting in no-click events require, to reach meaningful conclusions, the use of a fair-sampling assumption, a redefinition of the continuity parameter $\delta_0$, and a redefinition of the  quantity $\epsilon_{\mathrm{expt}}$. Second, we need an additional hypothesis on the epistemic model to ensure that the preparation used in the experiment yields (approximatively) pure coherent quantum states with a known overlap (see discussion in section~\ref{sec:mixed}). This is due to the fact that we did not explicitly check the actual performance of the preparation procedure by, e.g., performing a direct measurement of $\delta_0$. Such a verification would have required the use of a complex interferometer, and could have been performed only for very small dimension. Our approach was to use the very well understood physics of lasers operating well above threshold as a basis for modeling the quantum states produced by our preparation procedure. This allowed us to probe states in a much higher dimensional Hilbert space, and therefore small values of the continuity parameter $\delta_0$, than would have been possible otherwise.

Taking into account all the above constraints, our experiment nevertheless excludes, with a high degree of confidence, a large class of $\psi$-epistemic models. These $\psi$-epistemic models are labeled by two parameters, $\delta_0$ that describes how continuous the model is, and $\epsilon$ that describes how epistemic it is. Our experimental results exclude a region in the $\delta_0,\epsilon$ plane, see fig.~\ref{fig:epsilon2}.

Simultaneously with our work, an experimental test of $\psi$-epistemic models based on the original PBR no-go theorem using two ions in the same trap was reported in \cite{2012-Nigg}. Both experiments exclude (possibly different) classes of $\psi$-epistemic models. Both experiments require specific additional hypotheses to ensure that the prepared states are pure quantum states with desired properties (as discussed in the above paragraphs and in \cite{2012-Nigg}). Such additional hypotheses will probably be needed in any experimental test of $\psi$-epistemic models.

\begin{acknowledgments}
We thank anonymous referees for useful comments. We acknowledge financial support from the European Union under project QCS, from the FRS-FNRS under project DIQIP, and from the Brussels-Capital Region through a BB2B grant.
\end{acknowledgments}

\appendix

\section{Alternative treatment of the detection loophole}
\label{sec:alternative}

In this section we adopt an approach different from that followed in section~\ref{sec:detection} of the main text.

There are many different physical implementations of single-photon detectors which would give the same response to incoming light as the SSPD used in the experiment. Our main assumption is that for a given ontic state $\lambda$ the probabilities of outcomes are independent of how the detector is implemented, provided it gives the same measurement statistics as the actual one. This obviously constitutes a strong assumption on the behaviour of $\psi$-epistemic models. Its interest is that it allows us to keep the same continuity condition as in the main text, contrary to the reasoning in section~\ref{sec:detection}.

We suppose that the interaction of the detector with the electromagnetic field obeys photon-number superselection rule. That is, the excitations in the superconducting wire depend on the quanta of energy (or photon number) and not on the relative phases. With this hypothesis the response of the SSPD should be very well approximated by an ideal detector consisting of three steps:
\begin{enumerate}
\item a beam splitter, one arm of which absorbs a fraction $1-\eta$ of the coherent light, where $\eta$ is the efficiency of the actual detector;
\item the beam splitter is followed by a quantum non-demolition (QND) detector which measures the photon number (the result of this measurement is not provided to the experimenter);
\item the state at the output of the QND measurement is sent to a perfect single-photon detector with 100\% efficiency and dark-count rate $Dk$ which provides as output the first time bin in which a photon is detected.
\end{enumerate}
Such a model for a single-photon detector is standard in quantum optics \cite{1997-Leonhardt,2010-Kok}. Note that the actual SSPD may differ slightly from the above idealisation (for instance the attenuation may not be linear). Such small changes could be implemented by changing slightly the model, and do not affect our conclusions.

From the point of view of quantum theory, step 1 results in the preparation of an attenuated coherent state with photon mean number $\eta \langle n \rangle$, and step 2 results in the preparation of a mixture $\rho_k = \sum_n p_n \rho_k^n$ of states $\rho_k^n$ with photon number $n$, if the initial state was $|\alpha_k\rangle$. Note that the states $\rho_k^n$ are pure states. In particular the state $\rho_k^1 = |\psi_k\rangle\langle \psi_k|$, where $|\psi_k\rangle$ corresponds to the state given in the main text.

Now we make the assumption that, from the point of view of $\psi$-epistemic models, steps 1 and 2 can be thought of as additional state preparation steps. That is, the outcome of these steps is the preparation of a new ontic state $\lambda'$. Following \cite{2005-Spekkens}, we assume that the ontic distribution corresponding to a statistical mixture $\rho_k = \sum_n p_n \rho_k^n$ of states $\rho_k^n$  is of the form $\sum_n p_n P(\lambda|\rho_k^n)$. With these assumptions, we have the following identity for the predictions of $\psi$-epistemic models:
\begin{eqnarray*}
\sum_{k=1}^d P(k|\mathrm{clk},Q_{\alpha_k})
&=& \frac{1}{P(\mathrm{clk})}\sum_{k=1}^d \sum_{n=0}^\infty P(k|{\rho_k^n})P(n|{\alpha_k}) \\
&=& \frac{1}{P(\mathrm{clk})}\sum_{k=1}^d \sum_{n=0}^\infty \sum_\lambda P(k|\lambda)P(\lambda|{\rho_k^n}) \\
&\geq& \frac{p_1}{P(\mathrm{clk})} \sum_{k=1}^d \sum_\lambda P(k|\lambda)P(\lambda|{\rho_k^1}) \, ,
\end{eqnarray*}
where $n$ is the result of the QND measurement and $p_1=P(n=1|\alpha_k)$ is the probability for the QND measurement of detecting $n=1$ photons (in the last line we retain only the term $n=1$). Note that the $n=0$ component in line 2 only contributes through the effect of dark counts.

We now restrict our analysis to $\psi$-epistemic models associated to the single-photon states $\rho_k^1$. We assume $\delta$-continuity of $\psi$-epistemic models (using the definition given in the main text). If the states $\rho_k^1$ are $\delta$-close, then we have the bound $\sum_\lambda P(k|\lambda)P(\lambda|{\rho_k^1})>\epsilon>0$, as in the main text.

In summary the assumptions made above lead to the bound
\begin{equation*}
\sum_{k=1}^d P(k|\mathrm{clk},Q_{\alpha_k}) \geq \frac{p_1}{P(\mathrm{clk})}  \epsilon \, .
\end{equation*}
This is the same bound as obtained in the main text, except for the factor $\frac{p_1}{P(\mathrm{clk})} \simeq 0.90$ (where we take into account that $\langle n\rangle = 0.2$ and neglect dark counts).

\section{Analytic analysis of imperfect state preparation}
\label{sec:analytic}

As discussed in section~\ref{sec:mixed}, the presence of additional parameters $\varphi_j$ arising from phase fluctuations of the laser increases the distance $\delta_0$ with the reference state. The increase in distance is obtained from the following expressions (which we first give for the case where the state contains one photon):
\begin{eqnarray*}
|\psi \rangle &=& \frac{1}{\sqrt{d}} \sum_{j=1}^d |j\rangle \\
|\psi_{k ,\{\varphi_j\}} \rangle &=& \frac{1}{\sqrt{d-1}} \sum_{j\neq k}e^{i\varphi_j} |j\rangle \\
\left|E_\varphi\left[\langle \psi|\psi_{k,\{\varphi_j\}}\rangle \right]\right|^2 &\simeq& 1-\frac{1}{d}-\frac{d \Delta\varphi^2}{4} \, ,
\end{eqnarray*}
where $E_\varphi$ is the expectation over the random phases $\varphi$, and we give the leading order effect of the random phases (for $d$ large and $\Delta \varphi^2$ small). We thus find 
\begin{equation*}
E_\varphi\left[\delta\right] = 1-\sqrt{(d-1)/d}+\frac{d\Delta\varphi^2}{8} \, .
\end{equation*}
Now we insert this expression into eq.~(\ref{eq:deltatilde}) to obtain
\begin{eqnarray*}
E_\varphi\left[\delta(d,\alpha^2)\right]
&=& 1-E_\varphi\left[|\langle\tilde\alpha_{k,\{\varphi_j\}}|\tilde\alpha_0\rangle|\right] \\
&=& 1-\frac{e^{\alpha^2\sqrt{(d-1)/d}} \left(1-\alpha^2\frac{d \Delta\varphi^2}{8}\right) -1}{e^{\alpha^2}-1} \, .
\end{eqnarray*}
From this expression we see that $E\left[\delta\right]$ reaches a minimum when $d \simeq 40$. For larger $d$, the distance $\delta$ between the states is dominated by the random phases, rather than the presence or absence of time bin $k$. This is indeed what is observed by using the exact numerical method in the main text.

\end{document}